\documentclass[conference, letterpaper]{IEEEtran}
\ifCLASSINFOpdf
\else
\fi
\ifCLASSINFOpdf
   \usepackage[pdftex]{graphicx}
\else
\fi

\usepackage{amsmath,amsfonts,amssymb}
\usepackage{color}
\usepackage{fancyhdr}

\renewcommand{\thispagestyle}[2]{}

\fancypagestyle{plain}{
        \fancyhead{}
        \fancyhead[C]{first page center header}
        \fancyfoot{}
        \fancyfoot[C]{first page center footer}
}
\begin{document}

\newtheorem{theorem}{Theorem}
\newtheorem{proposition}[theorem]{Proposition}
\newtheorem{definition}[theorem]{Definition}
\newtheorem{remark}[theorem]{Remark}
\newtheorem{example}[theorem]{Example}
\newtheorem{claim}[theorem]{Claim}
\newtheorem{problem}[theorem]{Problem}
\newtheorem{corollary}[theorem]{Corollary}

%
\title{Searching for network modules}

\author{\IEEEauthorblockN{Giovanni Rossi}
\IEEEauthorblockA{Dept of Computer Science and Engineering - DISI\\
University of Bologna\\
40126 Bologna, Italy\\
Email: giovanni.rossi6@unibo.it}}


%


\maketitle

\begin{abstract}
The kinds of data meaningfully represented as networks seem ever-increasing, from the social sciences to PPI networks in bioinformatics. When analyzing these networks a key target is to uncover
their modular structure, which means searching for a family of modules, namely node subsets spanning each a subnetwork more densely connected than the average. Objective
function-based graph clustering procedures rely on optimization techniques such as modularity maximization, and output a partition of nodes, i.e. a family of pair-wise disjoint
subsets. But some nodes could be included in multiple modules, thus requiring an overlapping modular structure. The issue is mostly tackled by means of fuzzy clustering approaches, where
each node may be included in different modules with different [0,1]-ranged memberships.
This work proposes a novel type of objective function for graph clustering, in the form of a multilinear polynomial extension whose coefficients are determined by network topology. It may be
thought of as a potential function, to be maximized, taking its values on fuzzy clusterings or families of fuzzy subsets of nodes over which every node distributes a unit
membership. When suitably parameterized, this potential is shown to attain its maximum when every node concentrates its all unit membership on some module. The output thus remains a
partition, but the original discrete optimization problem is turned into a continuous version allowing to conceive alternative search strategies. The instance of the problem being
a pseudo-Boolean function assigning real-valued cluster scores to node subsets, modularity maximization is employed to exemplify a so-called quadratic form, in that the scores of 
singletons and pairs also fully determine the scores of larger clusters, while the resulting multilinear polynomial potential function has degree 2. After considering further
quadratic instances, different from modularity and obtained by interpreting network topology in alternative manners, a greedy local-search strategy for the continuous framework is
analytically compared with an existing greedy agglomerative procedure for the discrete case. Overlapping is finally discussed in terms of multiple runs, i.e. several local
searches with different initializations.   
\end{abstract}



%
\IEEEpeerreviewmaketitle

\section{Introduction}
Networks or graphs are pairs whose elements are a set of nodes or vertices and a set of unordered pairs of nodes, namely the links or edges. In
the weighted case, real-valued weights on edges generally measure some intensity of the relation between the two endvertices. A great variety of data may
be represented by means of networks \cite{BookBarabasiNewmann2006}, with weighted edges clearly allowing for more flexibility. An oldest and notorious example is found in the
social sciences, where nodes are individuals and links formalize friendship/influence relations. More recently, given the increasing availability of genome-scale data
on protein interactions, much attention is being paid to PPI (protein-to-protein interaction) networks, where nodes are proteins and links
formalize their interaction in metabolic processes \cite{
CFinderBioinfo2006,BMCBioinfoClusCoeff2007,PPIclusteringBioinformatics2007,PLOSCompBio2007,Proteomics2013,Proteins2004,MolSysBio2007,HierarLayerBioinfo2012,PPIclusteringBMCBioinfo2009}. 

These complex networks describing so different real-world phenomena display several common features \cite{ClusteringCoefficientInfoSciences20017}, and in particular are all
organized in modular structures
\cite{CommunitiesReviewNewmannSIAM2003,NewmanPNAS2006}. A module or community basically is a `heavy cluster' of nodes, in the following
sense. The subgraph spanned by a subset of vertices includes only these vertices in the subset and those edges whose endvertices are both in the subset. A module then is a vertex
subset spanning a subgraph whose edge set is exceptionally dense. In the weighted case, the edges of sugraphs spanned by modules are literally heavy, i.e. collectively receiving a
significant fraction of the total weight on all edges. Mathematically speaking, searching for families of network modules is a generalization of graph clustering
\cite{GraphClusteringSchaeffer2007}, where the goal is to determine a peculiar family of heavy vertex subsets, namely a partition \cite{Aigner97}. A partition of vertices is a
collection of non-empty and pair-wise disjoint vertex subsets, called blocks, whose union is the whole vertex set.

\subsection{Related work}
The present paper provides an optimization framework and a greedy local-search strategy for graph clustering and network module detection, with focus on fuzzy
and/or overlapping modular structures.
Objective function-based graph clustering methods determine partitions of vertices relying on optimization, i.e. maximizing or minimizing a score or a cost. Hence the blocks of the
generated partition are optimal or maximally heavy vertex subsets as they are found to maximize or minimize a global objective function, where this latter commonly is a so-called
additive partition function \cite{FortunatoLongSurvey2010}. This means that a suitable set function assigns a cluster score/cost to every subset
of vertices, and then the global score/cost of any partition simply obtains by summing the scores/costs of its blocks. A main example is modularity maximization
\cite{OnModularityClustering2007,NewmanFastAlgorithm2004,NewmanPNAS2006,LocalModularityClustering2011}, where in particular the scores of vertex subsets depend only on the scores of the $n$ singletons and the $\binom{n}{2}$ pairs, denoting by $n$ the total number of vertices.

In a partition every vertex is included in precisely one block, while network modules may well be non-disjoint or
overlapping \cite{LancichinettiEtAl2009,XieEtAl2012}, and this seems mostly relevant for PPI networks, where modules are protein complexes.
The search for overlapping modular structures often resorts to fuzzy clustering algorithms
\cite{FuzzyCmeansBook2008,FuzzyPotts2004,IEEENanoBio2012,RoughFuzzyPlosONE2014,ZhangWangZhang2007},
whose
output is a family of fuzzy vertex subsets. Hence vertices may be included in diffent modules, with different [0,1]-ranged memberships. One way to look at the whole framework
formalized in the sequel is to see it as a tool for extending additive partition functions over families of fuzzy subsets. In fact, as the proposed model conforms
any chosen objective function-based graph clustering method employing an additive partition function to the fuzzy setting, it may be step-wise exemplified for
modularity maximization, starting with Section II hereafter. Next Section III introduces pseudo-Boolean functions \cite{BorosHammer02} in terms of the cluster score of vertex
subsets, while Section IV formalizes the implications of evaluating fuzzy clusterings through the polynomial multilinear extension of set functions. Toward tight
comparison with modularity maximization, Section V defines two topology-based quadratic cluster score functions, with focus respectively on (i) the weighted case, and (ii) the
clustering coefficient of spanned subgraphs. Section VI details a greedy local-search strategy for the continuous framework, and compares it with a well-known fast heuristic for
modularity maximization. Overlapping is discussed in Subsection VI.B in terms of multiple local searches with different initializations, while the
conclusion in Section VII also briefly examines how to model random networks with overlapping modules, in view of future work.

\section{Modularity}
Denote by $N=\{1,\ldots,n\}$ the $n$-set of vertices and by $2^N=\{A:A\subseteq N\}$ the $2^n$-set of vertex subsets. A network is usually intended to be a simple graph $G=(N,E)$,
i.e. with edge set $E\subseteq N_2$ included in the $\binom{n}{2}$-set of unordered pairs of vertices: $N_2=\{\{i,j\}:1\leq i<j\leq n\}\subset2^N$, where
$(\cdot,\cdot)$ and $\{\cdot,\cdot\}$ respectively are ordered and unordered pairs, while $\subset$ is proper inclusion. Such an edge set $E$ is identified by its characteristic
function $\chi_E:N_2\rightarrow\{0,1\}$, defined by
$\chi_E(\{i,j\})=\left\{\begin{array}{c}1\text{ if }\{i,j\}\in E\\0\text{ if }\{i,j\}\in E^c=N_2\backslash E\end{array}\right.$, and thus Boolean vector
$\chi_E\in\{0,1\}^{\binom{n}{2}}$ is an extreme point of the $\binom{n}{2}$-dimensional unit hypercube $[0,1]^{\binom{n}{2}}$. Networks with [0,1]-ranged weights on
edges correspond to non-extreme points of this $\binom{n}{2}$-cube. In other terms, their edge set is fuzzy: it includes unordered pairs of vertices, each with a membership in
[0,1]. Thus weighted networks
$G=(N,W),W\in[0,1]^{\binom{n}{2}}$ comprise non-weighted ones as those $2^{\binom{n}{2}}$ special cases where all the $\binom{n}{2}$ weights
range in $\{0,1\}$. Let $w_{ij}$ denote the weight on pair $\{i,j\}\in N_2$ or equivalently the $ij$-th entry of vector $W$.

Modularity maximization is a fundamental approach to module detection in complex networks via objective function-based graph clustering, and is usually applied to the non-weighted
case. As already outlined, modularity is an additive partition function, meaning that it takes real values on families $P=\{A_1,\ldots,A_{|P|}\}$ of non-empty vertex subsets
$A_1,\ldots,A_{|P|}\in 2^N$, called blocks, such that $A_l\cap A_k=\emptyset$ for $1\leq l<k\leq|P|$ and $A_1\cup\cdots\cup A_{|P|}=N$, where $|\cdot|$ is the number
of elements or cardinality of a set. Modularity is denoted by $\mathcal Q$ and defined, for given network $G=(N,W)$, in terms of the following quantities:
$w_i=\sum_{j\in N\backslash i}w_{ij}$ for every vertex $i\in N$ and $w_N=\sum_{\{i,j\}\in N_2}w_{ij}$. In non-weighted networks $G=(N,\chi_E),E\subseteq N_2$ these quantities
respectively are the degree $d_i=w_i$ or number of neighbors of every vertex $i$ and the total number $|E|=w_N$ of edges, while $w_{ij}=a_{ij}$ is the $ij$-th entry of the
($n\times n$ symmetric and Boolean) adjacency matrix \cite{SpectraGraphsBook2011}. On any partition $P$ modularity $\mathcal Q$ takes value:\smallskip\\
$\text{ }\text{ }\mathcal Q(P)=\frac{1}{2|E|}\sum_{1\leq i,j\leq n}\left(a_{ij}-\frac{d_id_j}{2|E|}\right)\delta_P(i,j)$, with\smallskip\\
$\text{ }\text{ }\delta_P(i,j)=\left\{\begin{array}{c}1\text{ if }i,j\in A\text{ for a block }A\in P\text ,\\0\text{ otherwise.}\end{array}\right.$\\
%
%
Note that the sum is over the $n^2$ ordered pairs of vertices, including those $n$ of the form $(i,i),i\in N$, and $\delta_P(i,i)=1$ for all partitions $P$. Therefore,
$\mathcal Q(P)=\sum_{A\in P}v_{\mathcal Q}(A)=$
\begin{equation*}
=\sum_{A\in P}
\left[
\sum_{i\in A}\left(-\frac{d_i^2}{4|E|^2}\right)+
\sum_{\{i,j\}\subseteq A}\left(\frac{a_{ij}}{|E|}-\frac{d_id_j}{2|E|^2}\right)
\right]\text ,
\end{equation*}
and the relation between $\mathcal Q$ and set function $v_{\mathcal Q}:2^N\rightarrow\mathbb R$ defined inside square parenthesis shall soon be looked at from a combinatorial
perspective. This expression details how modularity is an additive partition function, since global score is the sum over blocks of their score. In fact, $v_{\mathcal Q}(A)$ is
a measure of cluster score obtained by comparison with a probabilistic model that can only be mentioned here for reasons of space. Apart from constant terms given by the scores of
singletons, $v_{\mathcal Q}(A)$ is essentially determined by the difference between the fraction $\sum_{\{i,j\}\subseteq A}a_{ij}/|E|$ of edges whose endpoints are both in $A$,
and the expectation $\sum_{\{i,j\}\subseteq A}d_id_j/(2|E|^2)$ of such a fraction in the so-called configuration model, namely the random graph with same degree sequence
$(d_i)_{i\in N}$ \cite[p. 200]{CommunitiesReviewNewmannSIAM2003}.

A key feature of additive partition functions such as modularity is that the $2^n$ values $(v_{\mathcal Q}(A))_{A\in 2^N}$ of the underlying cluster score function $v_{\mathcal Q}$
are fully determined by the $1+n+\binom{n}{2}$ values taken on the empty set, where obviously $v_{\mathcal Q}(\emptyset)=0$, on the $n$ singletons $\{i\},i\in N$, where
$v_{\mathcal Q}(\{i\})=-\frac{d_i}{4|E|^2}$, and on the $\binom{n}{2}$ pairs $\{i,j\}\in N_2$, where\smallskip\\
$\text{ }\text{ }v_{\mathcal Q}(\{i,j\})=\frac{a_{ij}}{|E|}-\frac{d_id_j}{2|E|^2}-\frac{d_i^2}{4|E|^2}-\frac{d_j^2}{4|E|^2}$.

An analog of the configuration model for weighted networks $G=(N,W)$ seems missing, thus the weighted version
\begin{equation*}
\mathcal Q(P)=\sum_{A\in P}
\left[
\sum_{i\in A}\left(-\frac{w_i^2}{4w_N^2}\right)+
\sum_{\{i,j\}\subseteq A}\left(\frac{w_{ij}}{w_N}-\frac{w_iw_j}{2w_N^2}\right)
\right]
\end{equation*}
of modularity may be an objective function for clustering weighted networks, but has no probabilistic interpretation.

\section{Pseudo-Boolean cluster score functions}
Subsets or clusters $A,B\in 2^N$ and partitions or clusterings $P,Q\in\mathcal P^N$ of vertex set $N$ are elements of two fundamental posets (partially ordered sets), respectively the Boolean lattice
$(2^N,\cap,\cup)$ of subsets of $N$ ordered by inclusion $\supseteq$ and the geometric lattice $(\mathcal P^N,\wedge,\vee)$ of partitions of $N$ ordered by
coarsening $\geqslant$, i.e. $P\geqslant Q$ if for every $B\in Q$ there is $A\in P$ such that $A\supseteq B$, where $\wedge$ and $\vee$ respectively denote
the `coarsest-finer-than' or meet and the `finest-coarser-than' or join operators \cite{Aigner97}. Since these posets are finite, lattice functions
$v:2^N\rightarrow\mathbb R$ and $V:\mathcal P^N\rightarrow\mathbb R$ may be dealt with as points $v\in\mathbb R^{2^n}$ and $V\in\mathbb R^{\mathcal B_n}$ in
vector spaces\footnote{$\mathcal B_k=\sum_{1\leq l\leq k}\mathcal S_{k,l}$ is the (Bell) number of partitions of a $k$-set, while $\mathcal S_{k,l}$ is the
Stirling number of the second kind, i.e. the number of partitions of a $k$-set into $l$ blocks \cite{Aigner97,ConcreteMathematics,Rota1964}.}. A fundamental
basis of these spaces (apart from the canonical one) is provided by the so-called zeta function $\zeta$, which works as follows: for every $A\in 2^N$ and every
$P\in\mathcal P^N$, define $\zeta_A:2^N\rightarrow\{0,1\}$ and $\zeta_P:\mathcal P^N\rightarrow\{0,1\}$ by
$\zeta_A(B)=\left\{\begin{array}{c}1\text{ if }B\supseteq A\\0\text{ otherwise}\end{array}\right .$ and
$\zeta_P(Q)=\left\{\begin{array}{c}1\text{ if }Q\geqslant P\\0\text{ otherwise}\end{array}\right .$ (for all $B\in 2^N,Q\in\mathcal P^N$).
Then, $\{\zeta_A:A\in 2^N\}$ is a basis of $\mathbb R^{2^n}$ and $\{\zeta_P:P\in\mathcal P^N\}$ is a basis of $\mathbb R^{\mathcal B_n}$ (with axes indexed
respectively by subsets $A$ and partitions $P$). Set functions $v$ and partition functions $V$ are linear combinations of the elements of these bases, with
coefficients $\mu^v(A),A\in 2^N$ and $\mu^V(P),P\in\mathcal P^N$. That is,
$v(\cdot)=\sum_{A\in 2^N}\zeta_A(\cdot)\mu^v(A)$  and $V(\cdot)=\sum_{P\in\mathcal P^N}\zeta_P(\cdot)\mu^V(P)$, or
$v(B)=\sum_{A\subseteq B}\mu^v(A)$ and $V(Q)=\sum_{P\leqslant Q}\mu^V(P)$.

Functions $\mu^v:2^N\rightarrow\mathbb R$ and $\mu^V:\mathcal P^N\rightarrow\mathbb R$ are the \textit{M\"obius inversions} \cite{Aigner97,Rota64}
respectively of $v$ and $V$, obeying recursions $\mu^v(A)=v(A)-\sum_{B\subset A}\mu^v(B)$ and analogously $\mu^V(P)=V(P)-\sum_{Q<P}\mu^V(Q)$,
where $P>Q$ denotes \textit{proper coarsening}: there exist at least two blocks $B,B'\in Q$ and a corresponding block $A\in P$ such that $A\supseteq(B\cup B')$.

A partition function $V$ is additive when some set function $v$ satisfies $V(P)=\sum_{A\in P}v(A)$ for all $P\in\mathcal P^N$, in which case M\"obius inversions $\mu^v$ and $\mu^V$
are of course related. In particular, $\mu^V$ takes value $0$ on all partitions apart (possibly) from those $2^n-n$ where the number of non-singleton blocks is $\leq 1$ \cite{GilboaLehrer90GG,GilboaLehrer91VI}, namely the \textit{modular elements}\footnote{Modularity $\mathcal Q$ is meant to evaluate modular structures in complex networks, while the modular elements of
$(\mathcal P^N,\wedge,\vee)$ are those partitions $\hat P$ realizing equality $r(\hat P\wedge Q)+r(\hat P\vee Q)=r(\hat P)+r(Q)$ for all $Q\in\mathcal P^N$, where $r(P)=n-|P|$ is
the rank (see \cite{Aigner97} on modular lattices/lattice functions).} \cite{Aigner97,Stanley1971}
of geometric lattice $(\mathcal P^N,\wedge,\vee)$, i.e. the top $P^{\top}=\{N\}$, the bottom $P_{\bot}=\{\{1\},\ldots,\{n\}\}$, and those with form $P^A_{\bot}=\{A,\{i_1\},\ldots ,\{i_{n-|A|}\}\}$ for $1<|A|<n$, where $\{i_1,\ldots ,i_{n-|A|}\}=N\backslash A=A^c$. Recursively, the values taken by
M\"obius inversion $\mu^V$ on these modular elements are:\smallskip\\
$(a)$ $\mu^V(P_{\bot})=\sum_{i\in N}v(\{i\})=n\mu^v(\emptyset)+\sum_{i\in N}\mu^v(\{i\})$,\smallskip\\
$(b)$ $\mu^V(P^A_{\bot})=\mu^v(A)+(-1)^{|A|+1}\mu^v(\emptyset)$ for $1<|A|\leq n$,\smallskip\\
with $P^N_{\bot}:=P^{\top}$. These partition functions admit a continuum of set functions $v,v'$ satisfying $\sum_{A\in P}v(A)=\sum_{A\in P}v'(A)$ for all $P$, the
requirements being $\sum_{i\in N}v({i})=\sum_{i\in N}v'({i})$ or $n\mu^v(\emptyset)+\sum_{i\in N}\mu^v({i})=n\mu^{v'}(\emptyset)+\sum_{i\in N}\mu^{v'}({i})$, as well as $\mu^{v'}(A)=v(A)-\sum_{B\subset A}\mu^{v'}(B)$ or $v'(A)=v(A)$ for $|A| > 1$. In general, $v(\emptyset)=\mu^v(\emptyset)$ and $v(\{i\})=\mu^v(\emptyset)+\mu^v(\{i\})$. If the set functions to be dealt with are meant to quantify the cluster score of vertex subsets, which is null for the empty set, then $v(\{i\})=\mu^v(\{i\})$ for all $i\in N$ as well as $\mu^v(\{i,j\})=v(\{i,j\})-v(\{i\})-v(\{j\})$ for all $\{i,j\}\in N_2$.

Geometrically, Boolean lattice $(2^N,\cap,\cup)$ is often dealt with as outlined above, namely as the set $\{0,1\}^n$ of extreme points of the $n$-cube $[0,1]^n$, in that subsets
$A\in 2^N$ correspond to their characteristic function $\chi_A:N\rightarrow\{0,1\}$, i.e. $\chi_A=(\chi_A(1),\ldots ,\chi_A(n))\in\{0,1\}^n$, where $\chi_A(i)=1$ if $i\in A$
and $\chi_A(j)=0$ if $j\in A^c$. In this view, set functions $v:2^N\rightarrow\mathbb R$ are pseudo-Boolean functions $\hat f^v:\{0,1\}^n\rightarrow\mathbb R$ \cite{BorosHammer02},
with the following polynomial MLE (multilinear extension) $f^v:[0,1]^n\rightarrow\mathbb R$ over the whole $n$-cube:
\begin{equation*}
f^v(q)=\sum_{A\in 2^N}\left(\prod_{i\in A}q_i\right)\mu^v(A)\text{, where }\prod_{i\in\emptyset}q_i:=1
\end{equation*}
and $q=(q_1,\ldots ,q_n)\in[0,1]^n$ is any fuzzy subset of $N$. This is an extension in that
$f^v(\chi_A)=\sum_{B\subseteq A}\mu^v(B)=v(A)$ at all extreme points $\chi_A\in\{0,1\}^n$. Polynomial $f^v(q)$ is multilinear in $n$ variables $q_1,\ldots,q_n$, with degree
$\max\{|A|:\mu^v(A)\neq 0\}$ and coefficients given by the non-zero values $\mu^v(A)\neq 0$ of M\"obius inversion. If $\mu^v(A)=0$ for all $A\in 2^N,|A|>1$, then $f^v$ is
\textit{linear} (and $v$ is a \textit{valuation} \cite{Aigner97} of Boolean lattice $(2^N,\cap,\cup)$, i.e. $v(A\cap B)+v(A\cup B)=v(A)+v(B)$ for all $A,B\in 2^N$).
Similarly, if $\mu^v(A)=0$ for $A\in 2^N,|A|>2$, then $f^v$ is \textit{quadratic}. If $\mu^v(\emptyset)=0$, then a linear $f^v$ satisfies $v(A)=\sum_{i\in A}\mu^v(\{i\})$, while a quadratic $f^v$ satisfies

$v(A)=\sum_{i\in A}\mu^v(\{i\})+\sum_{\{i,j\}\subseteq A}\mu^v(\{i,j\})$.

The MLE $f^{v_{\mathcal Q}}$ of cluster score function $v_{\mathcal Q}$ defining modularity $\mathcal Q(P)=\sum_{A\in P}v_{\mathcal Q}(A)$ is quadratic, with coefficients
$\mu^{v_{\mathcal Q}}(\{i\})=v_{\mathcal Q}(\{i\})=-[d_i/(2|E|)]^2$ for singletons $\{i\}$ and $\mu^{v_{\mathcal Q}}(\{i,j\})=\left[a_{ij}-d_id_j/(2|E|)\right]/|E|$ for pairs
$\{i,j\}$. Concerning conditions $(a)-(b)$ above, let $v_{\mathcal Q}'$ be an alternative cluster score function with quadratic MLE satisfying
$v_{\mathcal Q}'(\{i\})=\sum_{j\in N}v_{\mathcal Q}(\{j\})/n$ for all vertices $i\in N$. This means that in $v_{\mathcal Q}$ every vertex has its own
score when considered as a singleton cluster, while in $v_{\mathcal Q}'$ all vertices score the same as singletons. However, condition $(a)$ holds since
$\sum_{i\in N}v_{\mathcal Q}'(\{i\})=\sum_{i\in N}-d_i^2/(4|E|^2)=\sum_{i\in N}v_{\mathcal Q}(\{i\})$. By setting $\mu^{v_{\mathcal Q}'}(\{i,j\})=v_Q(\{i,j\})-\mu^{v'_{\mathcal Q}}(\{i\})-\mu^{v'_{\mathcal Q}}(\{j\})$ for
pairs and $\mu^{v_{\mathcal Q}'}(A)=0=\mu^{v_{\mathcal Q}}(A)$ for all subsets $A,1\neq|A|\neq2$, condition $(b)$ holds too, since the \textit{net added} score $\mu^{v_{\mathcal Q}'}(\{i,j\})=v_Q(\{i,j\})-\mu^{v'_{\mathcal Q}}(\{i\})-\mu^{v'_{\mathcal Q}}(\{j\})$ of pairs over singletons yields $v_Q'(\{i,j\})=v_Q(\{i,j\})$, and thus
$\sum_{A\in P}v_{\mathcal Q}'(A)=\mathcal Q(P)=\sum_{A\in P}v_{\mathcal Q}(A)$ for all $P\in\mathcal P^N$.

\section{Fuzzy clustering}
Denote by $2^N_i=\{A:i\in A\in 2^N\}$ the $2^{n-1}$-set consisting of all subsets where each $i\in N$ is included, and by $\Delta_i$ the associated
$2^{n-1}-1$-dimensional unit simplex whose extreme points are indexed by these subsets $A\in 2^N_i$. Slightly modifying the notation introduced in Section III, from now on let
$q_i\in\Delta_i$ be a generic membership distribution, with $q_i^A\in[0,1]$ quantifying $i$'s membership in cluster $A$. That is to say, $q_i:2^N_i\rightarrow[0,1]$ with
$q_i(A)=q_i^A$ and $\sum_{A\in2^N_i}q_i^A=1$. 
\begin{definition}
A fuzzy cover $\textbf q=\{q^A:A\in 2^N\}$ is a collection of $2^n$ fuzzy clusters $q^A=(q^A_1,\ldots,q^A_n)\in[0,1]^n$, where $q_i^A\in[0,1]$ if $i\in A$ and
$q^A_j=0$ if $j\in A^c$, while $\sum_{A\in 2^N_i}q^A_i=1$ for all $i\in N$.
\end{definition}
Apart from zero entries, fuzzy covers \textbf q thus essentially correspond to $n$-tuples $(q_1,\ldots,q_n)\in\underset{i\in N}{\times}\Delta_i$ of
membership distributions \cite{LancichinettiEtAl2009,NepuszEtAl2008}. The originality of the present contribution develops from evaluating fuzzy network modules through
the MLE $f^v$ of an underling cluster score function $v$, entailing that fuzzy covers $\textbf q=\{q^A:A\in 2^N\}$ attain additive global score $F^V(\textbf q)$ given by the sum
of the
$2^n$ values taken by $f^v$, i.e.
\begin{equation}
F^V(\textbf q)=\sum_{A\in 2^N}f^v(q^A)=\sum_{A\in 2^N}\sum_{B\supseteq A}\left(\prod_{i\in A}q_i^B\right)\mu^v(A)\text .
\end{equation}
In pseudo-Boolean optimization \cite{BorosHammer02}, the goal is to minimize or maximize a pseudo-Boolean function $\hat f^w:\{0,1\}^n\rightarrow\mathbb R$, i.e.
a set function $v:2^N\rightarrow\mathbb R$, with MLE $f^v:[0,1]^n\rightarrow\mathbb R$ thus allowing to turn several discrete optimization problems into a
continuous setting. In near-Boolean optimization \cite{Rossi2017}, the objective function has the form of $F^V(\textbf q)$ above, and the MLE
allows to deal with discrete optimization problems involving additive partition functions (namely maximum-weight set partitioning/packing) into a continuous setting. 
\begin{definition}
A fuzzy clustering is a fuzzy cover $\textbf q$ satisfying $|\{i:q^A_i>0\}|\in\{0,|A|\}$ for all $A\in 2^N$.
\end{definition}
Thus in a fuzzy clustering (or fuzzy partition), for all $A\in 2^N$, the number of those $i\in A$ with strictly positive membership $q^A_i>0$ is either $0$ or else $|A|$. As
shown below, the set of values taken by $F^V$ on fuzzy covers coincides with the set of values taken (solely) on fuzzy clusterings. 
\begin{proposition}
For any set function $v$, the range of $F^V$ is saturated by the values taken on fuzzy clusterings.
\end{proposition}
\begin{proof}
In a fuzzy cover $\textbf q=\{q^A:A\in 2^N\}$, for some ($\supseteq$-minimal) $A\in 2^N$, let $A^+_{\textbf q}=\{i:q_i^A>0\}$ satisfy\smallskip\\
$\emptyset\subset A^+_{\textbf q}\subset A$, i.e. $0<|A^+_{\textbf q}|=\alpha<|A|$. Then,

$F^V(\textbf q)=\sum_{B\in2^{A^+_{\textbf q}}}f^v(q^B)+\sum_{A'\in 2^N\backslash 2^{A^+_{\textbf q}}}f^v(q^{A'})$,\smallskip\\
where $2^A=\{B:B\subseteq A\}$ for all $A\in 2^N$. In particular,
\begin{equation*}
f^v(q^A)=\sum_{B\in2^{A^+_{\textbf q}}}\left(\prod_{i\in A^+_{\textbf q}}q_i^A\right)\mu^v(B)\text .
\end{equation*}
Consider another fuzzy cover $\hat{\textbf q}$ such that $\hat q^{A'}=q^{A'}$ for all $A'\in2^N\backslash2^{A^+_{\textbf q}}$, while $\hat q^A_i=0$ for all
$i\in A$, with group membership $q^A_A=\sum_{i\in A}q_i^A=\sum_{i\in A^+_{\textbf q}}q_i^A$ redistributed over subsets $B\in2^{A^+_{\textbf q}}$ according to:
\begin{equation*}
\sum_{B\in(2^N_i\cap 2^{A^+_{\textbf q}})}\hat q_i^B=q_i^A+\sum_{B\in(2^N_i\cap 2^{A^+_{\textbf q}})}q_i^B\text{ for all }i\in A^+_{\textbf q}\text ,
\end{equation*}
\begin{equation*}
\prod_{i\in B}\hat q_i^B=\prod_{i\in B}q_i^B+\prod_{i\in B}q_i^A\text{ for all }B\in 2^{A^+_{\textbf q}},|B|>1\text .
\end{equation*}
These $2^\alpha-1$ equations with $\sum_{1\leq k\leq\alpha}k\binom{\alpha}{k}>2^{\alpha}$ variables $\hat q_i^B,\emptyset\neq B\in2^{A^+_{\textbf q}}$ admit a
continuum of solutions or fuzzy covers $\hat{\textbf q}$ where $\sum_{B\in 2^{A^+_{\textbf q}}}f^v(\hat q^B)=$

$=f^v(q^A)+\sum_{B\in 2^{A^+_{\textbf q}}}f^v(q^B)\Rightarrow F^V(\textbf q)=F^V(\hat{\textbf q})$.\smallskip\\
When reiterated for all (if any) $A'\in2^N\backslash2^{A^+_{\textbf q}}$ where inequalities $0<|\{i:q_i^{A'}>0\}|<|A'|$ both hold, this procedure yields a final fuzzy clustering
$\hat{\textbf q}^*$ satisfying $F^V(\textbf q)=F^V(\hat{\textbf q}^*)$.
\end{proof}

\begin{example}
Let $A=\{1,2,\ldots\}\supset A^+_{\textbf q}=\{1,2\}$, hence

$f^v(q^A)=q_1^A\mu^v(\{1\})+q_2^A\mu^v(\{2\})+q_1^Aq_2^A\mu^v(\{1,2\})$,\smallskip\\
with the following three conditions for $\hat{\textbf q}$
\begin{itemize}
\item $\hat q_1^{\{1,2\}}+\hat q_1^{\{1\}}=q_1^{\{1,2\}}+q_1^{\{1\}}+q_1^A$,
\item $\hat q_2^{\{1,2\}}+\hat q_2^{\{2\}}=q_2^{\{1,2\}}+q_2^{\{2\}}+q_2^A$,
\item $\hat q_1^{\{1,2\}}\hat q_2^{\{1,2\}}=q_1^{\{1,2\}}q_2^{\{1,2\}}+q_1^Aq_2^A$,
\end{itemize}
and four variables $\hat q_1^{\{1\}},\hat q_1^{\{1,2\}},\hat q_2^{\{2\}},\hat q_2^{\{1,2\}}$. One solution is

$\hat q_1^{\{1,2\}}=\hat q_2^{\{1,2\}}=\sqrt{q_1^{\{1,2\}}q_2^{\{1,2\}}+q_1^Aq_2^A}>0$,

$\hat q_1^{\{1\}}=q_1^{\{1,2\}}+q_1^{\{1\}}+q_1^A-\sqrt{q_1^{\{1,2\}}q_2^{\{1,2\}}+q_1^Aq_2^A}>0$,

$\hat q_2^{\{2\}}=q_2^{\{1,2\}}+q_2^{\{2\}}+q_2^A-\sqrt{q_1^{\{1,2\}}q_2^{\{1,2\}}+q_1^Aq_2^A}>0$.
\end{example}

A main advantage of fuzzy clusters over hard ones is that they may display non-empty pair-wise intesections while also maintaining a unit (cumulative) membership
that every $i\in N$ distributes over $2^N_i$ \cite{FuzzyPotts2004,XieEtAl2012,ZhangWangZhang2007}. Yet, if fuzzy clusterings are evaluated via MLE as in
expression (1), then they cannot score better than hard ones or partitions $P=\{A_1,\ldots ,A_{|P|}\}$, where these latter correspond to
$2^n$-collections
$\textbf p=\{p^A:A\in 2^N\}$ defined by $p^A=\left\{\begin{array}{c}\chi_A\text{ if }A\in P\\\textbf 0\text{ if }A\in2^N\backslash P\end{array}\right .$,
with $\textbf 0\in\{0,1\}^n$ denoting the all-zero $n$-vector. Apart from zero entries, \textbf p coincides with the collection
$(\chi_{A_1},\ldots,\chi_{A_{|P|}})$ of the characteristic functions of $P$'s blocks, which are pair-wise disjoint extreme points of the $n$-cube, i.e.
$\langle\chi_{A_l},\chi_{A_k}\rangle=0$ for $1\leq l<k\leq|P|$, satisfying $\chi_1+\cdots+\chi_{A_{|P|}}=\chi_N=\textbf 1$, where $\langle\cdot,\cdot\rangle$
denotes scalar product and $\textbf 1\in\{0,1\}^n$ is the all-one $n$-vector. Expression (1) evaluates partitions $P\in\mathcal P^N$ as these
collections $\textbf p\subset\{0,1\}^n$ of disjoint extreme points of the $n$-cube:
$F^V(\textbf p)=\sum_{A\in 2^N}f^v(p^A)=\sum_{A\in P}f^v(\chi_A)=\sum_{A\in P}v(A)$.
\begin{proposition}
For any fuzzy clustering \textbf q and set function $v$, some partitions $P,P'$ satisfy
$F^V(\textbf p)\geq F^V(\textbf q)\geq F^V(\textbf p')$.
\end{proposition}

\begin{proof}
By isolating the contribution of membership $q_i$ to objective function $F^V(\textbf q)=F^V(q_i|\textbf q_{-i})$ when all other $n-1$ memberships
$q_j,j\neq i$ are given,
\begin{equation}
F^V(\textbf q)=F^V_i(q_i|\textbf q_{-i})+F^V_{-i}(\textbf q_{-i})\text ,
\end{equation}
where $F^V(\textbf q)=\sum_{A\in 2^N_i}f^v(q^A)+\sum_{A'\in 2^N\backslash 2^N_i}f^v(q^{A'})$,
\begin{equation*}
F^V_i(q_i|\textbf q_{-i})=\sum_{A\in 2^N_i}q^A_i\left[\sum_{B\subseteq A\backslash i}\left(\prod_{j\in B}q^A_j\right)\mu^v(B\cup i)\right]
\end{equation*}
\begin{equation*}
\text{and }\text{ }F^V_{-i}(\textbf q_{-i})=\sum_{A\in 2^N_i}\left[\sum_{B\subseteq A\backslash i}\left(\prod_{j\in B}q^A_j\right)\mu^v(B)\right]+
\end{equation*}
\begin{equation*}
+\sum_{A'\in 2^N\backslash 2^N_i}\left[\sum_{B'\subseteq A'}\left(\prod_{j'\in B'}q^{A'}_{j'}\right)\mu^v(B')\right]\text .
\end{equation*}
Define $v_{\textbf q_{-i}}:2^N_i\rightarrow\mathbb R$ by
\begin{equation}
v_{\textbf q_{-i}}(A)=\sum_{B\subseteq A\backslash i}\left(\prod_{j\in B}q^A_j\right)\mu^v(B\cup i)\text .
\end{equation}
Let $\mathbb A^+_{\textbf q_{-i}}=\arg\max v_{\textbf q_{-i}}$ and $\mathbb A^-_{\textbf q_{-i}}=\arg\min v_{\textbf q_{-i}}$,
where surely $\emptyset\subset\mathbb A^+_{\textbf q_{-i}},\mathbb A^-_{\textbf q_{-i}}\subseteq 2^N_i$. Most importantly,
\begin{equation}
F^V_i(q_i|\textbf q_{-i})=\sum_{A\in 2^N_i}\Big(q^A_i\cdot v_{\textbf q_{-i}}(A)\Big)=\langle q_i,v_{\textbf q_{-i}}\rangle\text .
\end{equation}
In words, for given membership distributions $q_j,j\neq i$, global score is affected by $i$'s membership distribution $q_i$ through a scalar product. In order to
maximize (or minimize) $F^V$ by suitably choosing $q_i$ for given $\textbf q_{-i}$, the whole of $i$'s membership mass has to be placed over
$\mathbb A^+_{\textbf q_{-i}}$ (or $\mathbb A^-_{\textbf q_{-i}}$), anyhow. Hence there are precisely $|\mathbb A^+_{\textbf q_{-i}}|>0$ (or
$|\mathbb A^-_{\textbf q_{-i}}|>0$) available extreme points of $\Delta_i$. After reiteration for all $i\in N$, the outcome shall generally consist of two fuzzy
covers $\overline{\textbf q}$ and $\underline{\textbf q}$ such that $F^V(\overline{\textbf q})\geq F^V(\textbf q)\geq F^V(\underline{\textbf q})$ as well as
$\overline q_i,\underline q_i\in ex(\Delta_i)$, where $ex(\Delta_i)$ is the $2^{n-1}$-set of extreme points of simplex $\Delta_i$. When this is combined with
Proposition 3, the desired conclusion follows.
\end{proof}

These findings suggest to search for optimal partitions through reiterated improvements of the objective function:
$F^V(\textbf q(t+1))>F^V(\textbf q(t)),t=0,1,\ldots$, while only requiring the cluster score function $v$ and an initial fuzzy clustering $\textbf q(0)$ as inputs. Before
considering how to search, two further topology-based quadratic cluster scores are defined hereafter.

\section{Weights and common neighbors}
Toward comparison with modularity, a first quadratic cluster score function obtains by focusing on weighted networks, where it seems natural to set the
score $v(\{i,j\})$ of pairs equal to edge weights $w_{ij}$. Since $\emptyset$ scores 0, the quadratic form then only requires to define the score $v(\{i\})$ of singletons.
Basically, the greater $w_i=\sum_{j\in N\backslash i}w_{ij}$, the smaller the score of $i$ as a singleton
cluster. To formalize this, attention is placed on the complement or dual edge set $\overline W\in[0,1]^{\binom{n}{2}}$, whose entries $\overline{w}_{ij}=1-w_{ij}$ measure a
repulsion between $i$ and $j$.\\ A second quadratic cluster score function is defined by focusing on the clustering coefficient (of non-weighted networks). A distinctive
feature of several complex networks (including social ones) is an high density of `triangles', namely triples of nodes any two of which are linked. Such a density is measured in
terms of a [0,1]-ranged ratio by the so-called clustering coefficient, which is in fact the probability that by picking at random a vertex and two of its neighbors these latter are
also neighbors of each other \cite{ClusteringCoefficientInfoSciences20017,NewmanSocialAreDifferent2003}. A cluster score function can thus be conceived to measure the density of
triangles \textit{locally}, namely in the subgraphs spanned by vertex subsets. To this end, the score of any pair of vertices shall be determined by counting their common neighbors
\cite{Ahn2010,XieEtAl2012}.

\subsection{Score of singletons and dual weights}
Given a weighted network $G=(N,W)$, consider the cluster score function $v$ defined by $v(\emptyset)=0$ and $v(\{i,j\})=w_{ij}$ for pairs, while for singletons
\begin{equation}
v(\{i\})=\sum_{j\in N\backslash i}\frac{1-w_{ij}}{2(n-1)}=\frac{n-1-w_i}{2(n-1)}
\end{equation}
as well as $\mu^v(A)=0$ for larger subsets $A,|A|>2$. Weight $w_{ij}$ and its dual $\overline{w}_{ij}=1-w_{ij}$ may measure
respectively an attraction and a repulsion, while these quadratic cluster scores load the former on pairs and the latter on singletons. In particular, $\overline{w}_{ij}$ is
equally shared between $i$ and $j$, and the score of each singleton is the arithmetic mean of its $n-1$ shares. Therefore, $v(\{i\})\in[0,\frac{1}{2}]$ attains the upper bound on
isolated vertices, namely those $i$ such that $w_i=0$, and the lower one on those $i$ such that $w_i=n-1$. The  values of M\"obius inversion on pairs are
\begin{equation*}
\mu^v(\{i,j\})=w_{ij}-2\frac{\overline{w}_{ij}}{2(n-1)}-\sum_{k\in N\backslash\{i,j\}}\frac{2-w_{ik}-w_{jk}}{2(n-1)}=
\end{equation*}
\begin{equation}
=\frac{w_i+w_j}{2(n-1)}-\overline{w}_{ij}=w_{ij}-\left(1-\frac{w_i+w_j}{2(n-1)}\right)\text .
\end{equation}

Cluster score function $v$, with M\"obius inversion $\mu^v$ taking non-zero values only on singletons and pairs according to expressions (5-6), may be compared with
modularity-based $v_{\mathcal Q}$ by focusing on simple graphs $G=(N,\chi_E),E\subseteq N_2$. With the notation introduced in Section II,
$v(\{i\})=\frac{n-1-d_i}{2(n-1)}\geq0$ and $v_{\mathcal Q}(\{i\})=\frac{-d_i}{4|E|^2}\leq0$ on singletons, while on pairs
$\mu^v(\{i,j\})=a_{ij}-1+\frac{d_i+d_j}{2(n-1)}$ and $\mu^{v_{\mathcal Q}}(\{i,j\})=\frac{a_{ij}}{|E|}-\frac{d_id_j}{2|E|^2}$. Hence if $a_{ij}=1$ 
then $\mu^v(\{i,j\}),\mu^{v_{\mathcal Q}}(\{i,j\})\geq 0$, while if $a_{ij}=0$ then $\mu^v(\{i,j\}),\mu^{v_{\mathcal Q}}(\{i,j\})\leq 0$. It may be noted that for singletons both
$\mu^v$ and $\mu^{v_{\mathcal Q}}$ assign maximum/minimum cluster score to vertices $i$ of lowest/highest degree $d_i$.

Denoting by $d_A=\sum_{i\in A}d_i$ the group degree for $A\in 2^N$, expressions (5-6) define the score of larger clusters by

$v(A)=\frac{|A|}{2}+\frac{d_A(|A|-2)}{2(n-1)}-\left(\binom{|A|}{2}-|E(A)|\right)$,\smallskip\\
where $E(A)=\{\{i,j\}:E\ni\{i,j\}\subseteq A\}$ is the edge set of the subgraph $G(A)=(A,E(A))$ spanned or induced\footnote{The terminology and notation used here are standard in
graph theory \cite{DiestelGraphTheory2010}.} by $A$. The first two summands are evidently rather rough, as they entail that many vertices of high degree constitute a valuable
cluster, independently from how many of them are adjacent. However, the third summand is precisely the number of edges that spanned subgraph $G(A)$ lacks with respect to the
complete one $K_A$. In particular, if $G(A)=K_A$ is complete and also a component (or maximal connected subgraph) of $G$, then $|E(A)|=\binom{|A|}{2}$ and $d_A=|A|(|A|-1)$.
Substituting,

$v(A)=\frac{|A|}{2}+\binom{|A|}{2}\frac{|A|-2}{n-1}$\smallskip\\
is thus the maximum that a $|A|$-subset of vertices may score in any network. Coming to partitions, if $G=K_N$ is the complete graph, then the additive partition function $V$
resulting from $v$ attains its unique maximum on the coarsest or top partition $P^{\top}$, where $V(P^{\top})=v(N)=\binom{n}{2}$. In the opposite case, if $G=(N,\emptyset)$ is the
empty graph, then additive global score $V$ attains its unique maximum on the finest or bottom partition $P_{\bot}$, where
$V(P_{\bot})=\sum_{i\in N}v(\{i\})=\frac{n}{2}$. In combinatorial theory, partitions $P=\{A_1,\ldots,A_{|P|}\}$ of $N$ correspond to those graphs
$G_P=K_{A_1}\cup\cdots\cup K_{A_{|P|}}$ on $N$ each of whose components is complete (partitions being elements of the so-called polygon matroid on the edges of the complete graph
of order $n$ \cite{Aigner97}). In terms of graph clustering problems, these $\mathcal B_n$ partition-like graphs clearly constitute the simplest conceivable instances, in that
global score attains its unique maximum $V(P)=\sum_{A\in P}\left(\frac{|A|}{2}+\binom{|A|}{2}\frac{|A|-2}{n-1}\right)$ precisely on the corresponding partition.

\subsection{Score of pairs and common neighbors}
The clustering coefficient of a network $G=(N,E)$ is
\begin{equation*}
cc(G)=\frac{3\times\text{ number of included cycles on 3 vertices}}{\text{number of included trees on 3 vertices}}\text ,
\end{equation*}
cycles and trees on 3 vertices \cite{DiestelGraphTheory2010} being also termed respectively `triangles' and `connected triples' \cite{NewmanSocialAreDifferent2003}.
Three vertices $i,j,k$ spanning a complete subgraph $G(\{i,j,k\})=K_{\{i,j,k\}}$, i.e. a cycle (of length 3), provide 3 trees, while if $G(\{i,j,k\})$ is connected but not complete then there
is only one tree. A disconnected $G(\{i,j,k\})$ evidently provides no tree. In network analysis $cc(G)$ is a key indicator measuring `transitivity', namely to what extent sharing
some common neighbor entails being adjacent, for any two vertices. While in social networks the clustering coefficient is higher than in non-social ones
\cite{NewmanSocialAreDifferent2003}, several complex networks display the same asymptotic clustering
coefficient as certain strongly regular graphs, in contrast to small-world networks \cite{BollobasScaleFree2003,ClusteringCoefficientInfoSciences20017}.

Now consider the aim to assign scores $v(A)$ to clusters $A$ in a way such that
higher values of the clustering coefficient $cc(G(A))$ over spanned subgraphs provide greater scores. A natural way to achieve this is by means of a \textit{cubic}
pseudo-Boolean function, i.e. such that $\mu^v(A)=0$ if $|A|>3$, which also seems to best interpret the model of random graphs with clustering where edges within trtriangles are
enumerated separately from other edges \cite{NewmanRandomTriangles2009}. Maintaining the
$\binom{n+1}{2}$ values on singletons and pairs as given by expressions (5-6) with weights $w_{ij}\in\{0,1\}$, on the $\binom{n}{3}$
3-subsets $\{i,j,k\}\in 2^N$ M\"obius inversion may be defined simply by $\mu^v(\{i,j,k\})=$\smallskip\\
$=\left\{\begin{array}{c}\beta\text{ if }G(\{i,j,k\})=K_{\{i,j,k\}}\text{ is complete,}\\
0\text{ if }G(\{i,j,k\})\text{ is connected but not complete,}\\
-\beta\text{ if }G(\{i,j,k\})\text{ is disconnected,}\end{array}\right.$\smallskip\\
with $\beta\in(0,1]$. The resulting cluster scores incorporate additional reward/penalty for proximity/distance to/from completeness of the spanned subgraph $G(A)$, and if
$G(A)=K_A$ is both complete and a component of $G$ then\smallskip\\
$\text{ }\text{ }v(A)=\frac{|A|}{2}+\binom{|A|}{2}\frac{|A|-2}{n-1}+\beta\binom{|A|}{3}$.\smallskip\\
However, the same target can be achieved by means of a quadratic $f^v$ where the count of both common and non-common neighbors for any two vertices determines the values taken by
M\"obius inversion $\mu^v$ on pairs. Consider the neighborhood of vertex $i$ in the given network $G=(N,E)$
\cite{Ahn2010,XieEtAl2012}, i.e. $\mathcal N_i=\{j:\{i,j\}\in E\}$.
Then, $|\mathcal N_i\cap\mathcal N_j|$ is the number of neighbors common to both $i$ and $j$,
while symmetric difference
$\mathcal N_i\Delta\mathcal N_j=(\mathcal N_i\backslash\mathcal N_j)\cup(\mathcal N_j\backslash\mathcal N_i)$ contains non-common neighbors. Quadratic scores can be defined by
$\mu^v(A)=0$ if $1\neq|A|\neq2$, while
on singletons and pairs
\begin{eqnarray}
\mu^v(\{i\})&=&\frac{1}{1+|\mathcal N_i|}\text{ }\text{ }\text{ }\text{ }\Big(=v(\{i\}\Big),\\
\mu^w(\{i,j\})&=&a_{ij}+\frac{|\mathcal N_i\cap\mathcal N_j|-|\mathcal N_i\Delta\mathcal N_j|}{|\mathcal N_i\cup\mathcal  N_j|}\text .
\end{eqnarray}
The resulting cluster score set function is
$$v(A)=|E(A)|+\sum_{i\in A}\frac{1}{1+d_i}+
\sum_{\{i,j\}\subseteq A}\frac{|\mathcal N_i\cap\mathcal N_j|-|\mathcal N_i\Delta \mathcal N_j|}{|\mathcal N_i\cup\mathcal  N_j|}\text ,$$
hence a spanned subgraph $G(A)=K_A$ which is complete and a component of $G$ scores $v(A)=(|A|-1)(|A|-2)+1$.

The remainder of this work focuses on searching for partitions $\textbf p\in\Delta_1\times\cdots\times\Delta_n$ that locally maximize objective function $F^V(\textbf q)$ in
expression (1), when the input consists of both: (i) $\binom{n+1}{2}$ values of M\"obius inversion $\mu^v$ on singletons and pairs, and (ii) an initial fuzzy clustering
$\textbf q(0)$. In particular, modularity-based $v_{\mathcal Q}$ together with these two cluster score functions defined respectively by expressions (5-6) and (7-8) shall be
furter compared as alternative inputs.

\section{Greedy search}
How to employ the results of Section IV for graph clustering may be detailed through comparison with the so-called greedy agglomerative approach to moduarity maximization
\cite{OnModularityClustering2007,NewmanFastAlgorithm2004}. Starting from the finest partition, this algorithm iteratively selects one union of two blocks that results in a maximal
increase of global score, thereby yielding a sequence $P(t+1)\gtrdot P(t)$ of partitions as the search path, where $P(0)=P_{\bot}$ and $\gtrdot$ denotes the covering relation
between partitions, i.e. $|P(t+1)|=|P(t)|-1$ and $P(t+1)>P(t)$ (see above). If there are tails, meaning that aletrnative unions of two blocks of $P(t)$ yield the same maximal
increase of global score, then the two blocks to be merged are randomly selected. The stopping criterion is the absence of any further improvement. The iterative procedure thus is
the following.\smallskip\\
\textsl{GreedyMerging}$(v,P)$\smallskip\\
\textsl{Initialize:} Set $t=0$ and $P(0)=P_{\bot}$.\smallskip\\
\textsl{Loop:} While $v(A\cup B)-v(A)-v(B)>0$ for some $A,B\in P(t)$, set $t=t+1$ and\smallskip\\
$[1]$ select (randomizing in case of tails) $A,B\in P(t-1)$ such that for all $A',B'\in P(t-1)$,
\begin{equation*}
v(A\cup B)-v(A)-v(B)\geq v(A'\cup B')-v(A')-v(B')\text ,
\end{equation*}
$[2]$ define $P(t)=\{A\cup B\}\cup(P(t-1)\backslash\{A,B\})$, i.e. obtain $P(t)$ from $P(t-1)$ by merging $A$ and $B$.\smallskip\\
\textsl{Output:} Set $P^*=P(t)$.

This algorithm has been used to maximize modularity in different types of networks \cite{NewmanFastAlgorithm2004}. In terms of combinatorial optimization, it is an heuristic
(for the NP-hard maximum-weight set partitioning problem), meaning that its worst-case output is not guaranteed to provide any bounded approximation of optimal global score. In
fact, for the $\frac{n}{2}$-regular graphs considered in \cite[Theorem 5.1]{OnModularityClustering2007}, \textsl{GreedyMerging} provides a worst-case solution $\hat P$
with zero modularity score $\mathcal Q(\hat P)=0$, while the optimal solution $P^*$ provides a strictly positive score $\mathcal Q(P^*)>0$. These graphs $G=(N,E)$ have
an even number $n>4$ of vertices, with two disjoint vertex subsets of equal size: $N^1=\{i_1,\ldots ,i_{\frac{n}{2}}\},N^2=\{j_1,\ldots ,j_{\frac{n}{2}}\}$. The edge set is
$E=E(K_{N^1})\cup E(K_{N^2})\cup\{\{i_k,j_k\}:1\leq k\leq\frac{n}{2}\}$, where $E(K_{N^1})=\{\{i,i'\}:\{i,i'\}\subset N^1\}$ and the same for $E(K_{N^2})$.
Hence $E$ includes the edges of complete graphs $K_{N^1}$ and $K_{N^2}$ together with all the $\frac{n}{2}$ edges with endpoints $i_k\in N^1$ and $j_k\in N^2$ for
$1\leq k\leq\frac{n}{2}$. Therefore, $|E|=2\binom{\frac{n}{2}}{2}+\frac{n}{2}=\frac{n^2}{4}$, with degree $d_i=\frac{n}{2}$ for all $i\in N$.
At $t=0$, for each of the $\binom{n}{2}$ unions of two blocks $\{i\},\{j\}\in P_{\bot}$, the corresponding variation $\mathcal Q(P(1))-\mathcal Q(P(0))$
of modularity equals $v_{\mathcal Q}(\{i,j\})-v_{\mathcal Q}(\{i\}-v_{\mathcal Q}(\{j\}=\mu^{v_{\mathcal Q}}(\{i,j\})=$\\
$=\frac{a_{ij}}{|E|}-\frac{d_id_j}{2|E|^2}=\left\{\begin{array}{c}2/n^2\text{ if }\{i,j\}\in E\text ,\\
-2/n^2\text{ if }\{i,j\}\in E^c\text .\end{array}\right.$ Hence the worst-case output is $\hat P=\{\{i_1,j_1\},\ldots ,\{i_{\frac{n}{2}},j_{\frac{n}{2}}\}\}$, i.e.
the partition obtained in $\frac{n}{2}$ iterations through unions $\{i_k\}\cup\{j_k\}$ for $1\leq k\leq\frac{n}{2}$, where modularity scores 0. Explicitely, $\mathcal Q(\hat P)=$\\
$=\sum_{i\in N}v_{\mathcal Q}(\{i\})+\sum_{1\leq k\leq\frac{n}{2}}\mu^{v_{\mathcal Q}}(\{i_k,j_k\})=$\\
$=-\sum_{i\in N}\frac{d_i^2}{4|E|^2}+\frac{n}{2}\frac{2}{n^2}=-\frac{1}{n}+\frac{1}{n}=0$. On the other hand, the unique maximum attains at $P^*=\{N^1,N^2\}$ where 
$\mathcal Q(P^*)=\sum_{i\in N}v_{\mathcal Q}(\{i\})+2\sum_{\{i,i'\}\subset N^1}\mu^{v_{\mathcal Q}}(\{i,i'\})=$\\
$=-\frac{1}{n}+2\binom{\frac{n}{2}}{2}\frac{2}{n^2}=\frac{n-4}{2n}>0$ as $n>4$.

For the other two quadratic scores defined in Section V, \textsl{GreedyMerging} provides the same worst-case output even when the input
cluster score function $v$ is that defined by expressions (5-6), in which case the $\binom{n}{2}$ unions of two blocks of $P_{\bot}$ result in a variation of global score equal to
$v(\{i,j\})-v(\{i\}-v(\{j\}=\mu^v(\{i,j\})=a_{ij}-1+\frac{d_i+d_j}{2(n-1)}=$\\ $=\left\{\begin{array}{c}\frac{n}{2(n-1)}\text{ if }\{i,j\}\in E\text ,\\
\frac{2-n}{2(n-1)}\text{ if }\{i,j\}\in E^c\text .\end{array}\right.$ 
But if the input is $v$ defined by expressions (7-8), then the algorithm surely finds the optimum $P^*$, as the $\binom{n}{2}$
unions of two blocks of $P_{\bot}$ result in variation $\mu^v(\{i,j\})=
a_{ij}+\frac{|\mathcal N_i\cap\mathcal N_j|-|\mathcal N_i\Delta\mathcal N_j|}{|\mathcal N_i\cup\mathcal N_j|}=$\\
$=\left\{\begin{array}{c}\frac{2(n-2)}{n+4}\text{ if }\{i,j\}\subset N^1\text{ or }\{i,j\}\subset N^2\text ,\\
\frac{2}{n}\text{ if }\{i,j\}\in E,i\in N^1,j\in N^2\text ,\\
\frac{6-n}{n-2}\text{ if }\{i,j\}\in E^c\text .\end{array}\right.$

Coming to fuzzy clusterings $\textbf q\in\Delta_1\times\cdots\times\Delta_n$, a quadratic $v$ reduces the objective function in expression (1) to $F^V(\textbf q)=$
\begin{equation}
=\sum_{i\in N}v(\{i\})+\sum_{\{i,j\}\in N_2}\left(\sum_{A\supseteq\{i,j\}}q_i^Aq_j^A\right)\mu^v(\{i,j\})\text .
\end{equation}
The main advantage of this MLE of additive partition functions $V$ is that it allows search paths to start and develop in the continuum $\Delta_1\times\cdots\times\Delta_n$,
although the output shall be a partition in view on Proposition 5. Designing a search strategy as a sequence $\textbf q(t)$ such that
$F^V(\textbf q(t+1))>F^V(\textbf q(t))$ amounts to formalize: an initial $\textbf q(0)$, how $\textbf q(t+1)$ obtains from the reached $\textbf q(t)$, and a stopping criterion.
By starting from an arbitrary input $\textbf q(0)$, the search is local, although the more the initial $n$ membership distributions
are each spread over $2^N_i$, the more it becomes global. The stopping criterion is determined by the definition of local optimality: if
$\mathcal N(\textbf q)=\underset{i\in N}{\cup}\{\hat q_i|\textbf q_{-i}:\hat q_i\in\Delta_i\}$ is the neighborhood of \textbf q, then $\textbf q^*$ is a local optimum
when $F^V(\textbf q^*)\geq F^V(\hat{\textbf q})$ for all $\hat{\textbf q}\in\mathcal N(\textbf q^*)$. In words, the neighborhood of \textbf q contains all
$n$-tuples of membership distributions where $n-1$ distributions are as in \textbf q while only one may vary, and $\textbf q^*$ is a local optimum if $F^V(\textbf q^*)$
is the greatest value taken by $F^V$ over $\mathcal N(\textbf q^*)$. It is shown below that for any partition $P$ a necessary and sufficient
condition for this local optimality is $v(A)\geq v(A\backslash i)+v(\{i\})$ for all $i\in A$ and all $A\in P$. A typical greedy local search would thus progress through a sequence
$\textbf q(t)$ such that $\textbf q(t+1)\in\mathcal N(\textbf q(t))$ and $F^V(\textbf q(t+1))-F^V(\textbf q(t))$ is maximal, but none of these two conditions is here maintained. In
fact, $\textbf q(t+1)\notin\mathcal N(\textbf q(t))$ as more than one of the $n$ membership distributions $(q_1(t),\ldots ,q_n(t))=\textbf q(t)$ shall vary within the same
$t$-th iteration, and rather than being applied directly to the increase $F^V(\textbf q(t+1))-F^V(\textbf q(t))$ of global score, greediness is applied to average derivatives,
defined hereafter.  

The $i$-th derivative \cite{BorosHammer02} of the MLE $f^v$ of $v$ at $x$ is $f^v_i(x)=$\smallskip\\
$=f^v(x_1,\ldots,x_{i-1},1,x_{i+1},\ldots,x_n)+$

$-f^v(x_1,\ldots,x_{i-1},0,x_{i+1},\ldots,x_n)=$
$$=\sum_{A\in 2_i^N}\left(\prod_{j\in A\backslash i}x_j\right)\mu^v(A)\text{, for } x=(x_1,\ldots,x_n)\in[0,1]^n\text .$$
At vertices $\chi_B,B\in 2^N$ of the $n$-cube it takes values $f^v_i(\chi_B)=\left\{\begin{array}{c}v(B)-v(B\backslash i)\text{ if }B\in 2^N_i\text{ , and}\\
v(B\cup i)-v(B)\text{ if }B\in2^N\backslash2^N_i\text .\end{array}\right.$
This derivative may be reproduced for the MLE $F^V$ of additive partition functions $V$ as follows. For all $i\in N$ and all $A\in2^N_i$, define membership distribution $q_{i_A}$
by $q_{i_A}^B=\left\{\begin{array}{c}1\text{ if }B=A\text ,\\ 0\text{ otherwise.}\end{array}\right.$ Also let $q_{i_{\emptyset}}^B=0$ for \textit{all}
$B\in2^N_i$, noting that $q_{i_{\emptyset}}$ is \textit{not} a membership distribution, as it places no membership over $2^N_i$ at all. Now define the $i_A$-derivative of $F^V$ at
\textbf q by $F^V_{i_A}(\textbf q)=$
\begin{equation*}
=F^V(q_{i_A}|\textbf q_{-i})-F^V(q_{i_{\emptyset}}|\textbf q_{-i})=F^V_i(q_{i_A}|\textbf q_{-i})=v_{\textbf q_{-i}}(A)\text ,
\end{equation*}
where the last two equalities obtain from expressions (2-3) in Section IV. Observe that if the $|A|-1$
membership distributions $q_j,j\in A\backslash i$ are $q_j^A=1$, then $F^V_{i_A}(\textbf q)=v(A)-v(A\backslash i)$, while
$F^V_{i_{\{i\}}}(\textbf q)=v(\{i\})$ independently from \textbf q. These $n2^{n-1}$ derivatives $(F^V_{i_A}(\textbf q(t)))_{i\in N,A\in 2^N_i}$ inform about
how to obtain $\textbf q(t+1)$ from the reached $\textbf q(t)$ in order to maximize the objective function. In particular, any greedy strategy requires
first to make clear what maximum distance may separate $\textbf q(t+1)$ from $\textbf q(t)$. Instead of $\textbf q(t+1)\in\mathcal N(\textbf q(t))$,
the rule maintained here is the same, \textit{mutatis mutandis}, as for \textsl{GreedyMerging}, namely that precisely one block is formed when transforming $\textbf q(t)$ into
$\textbf q(t+1)$. Hence $\sum_{i\in A}q^A_i(t)<|A|=\sum_{i\in A}q_i^A(t+1)$, or equivalently $q^A(t)\neq\chi_A=q^A(t+1)$, for exactly one $A$ at each $t$. Given this,
greediness is applied to \textit{average derivative}
$$\bar F^V_A(\textbf q)=\sum_{i\in A}\frac{v_{\textbf q_{-i}}(A)}{|A|}=
\sum_{B\subseteq A}\left[\sum_{i\in B}\left(\prod_{j\in B\backslash i}q^A_j\right)\right]\frac{\mu^v(B)}{|A|}\text .$$
In view of expression (9), for quadratic $f^v$ this reduces to 
\begin{equation*}
\bar F^V_A(\textbf q)=\frac{1}{|A|}\left[\sum_{i\in A}v(\{i\})+\sum_{\{i,j\}\subseteq A}\left(q_i^A+q_j^A\right)\mu^v(\{i,j\})\right]\text .
\end{equation*}
Thus the chosen $A$ (at iteration $t$, to be a block of the output partition $\textbf p^*$ being constructed) is one where $q^A(t)\neq\chi_A$ and
$\bar F^V_A(\textbf q(t))$ is maximal (dealing arbitrarily with tails). Then in the same iteration $t$ it is specified how to redistribute
membership $\sum_{B\in 2^N_j:B\cap A\neq\emptyset}q_j^B(t)$ over those $B\in 2^N_j$ such that $B\cap A=\emptyset$, for all $j\in A^c$.

Concerning the stopping criterion, the greedy procedure stops when $q^A(t)\in\{\textbf 0,\chi_A\}$, i.e. $\sum_{i\in A}q_i^A(t)\in\{0,|A|\}$, for all $A\in 2^N$. Thus, ignoring
zeros, $\textbf q(t)=\textbf p^*$ is a partition $P^*=\{A_1,\ldots,A_{|P^*|}\}$, dealt with in its Boolean representation $\textbf p^*=(\chi_{A_1},\ldots,\chi_{A_{|P^*|}})$. Next,
a second loop checks local optimality for this $P^*$, which attains if $v(A)\geq v(A\backslash i)+v(\{i\})$ for all $i\in A$ and all $A\in P^*$. If the inequality is not satisfied,
then the partition updates by splitting block $A$ in the two (new) blocks $A\backslash i$ and $\{i\}$. Finally, as for the starting point $\textbf q(0)$, there surely exist many
options, including a simplest (but computationally most demanding) one given by the $n$-tuple of uniform distributions $q_i^A(0)=2^{1-n}$ for all $A\in 2^N_i$ and all $i\in N$.
Broadly speaking, input $\textbf q(0)$ sets the terms of trade between computational burden and search width, as the more distributions $q_i(0),i\in N$ are each spread over
$2^N_i$, the more computationally demanding and wider becomes the search. Specifically, if a family $\mathcal F=\{A_1,\ldots,A_k\}\subset 2^N$ satisfies $q_i^B(0)=0$ for all
$B\in 2^N\backslash(2^{A_1}\cup\cdots\cup2^{A_k})$ and all $i\in N$ as well as $q^{A_l}_i(0)\neq 0$ for all $i\in A_l,1\leq l\leq k$, then the algorithm
proposed hereafter searches for optimal blocks only inside $2^{A_1}\cup\cdots\cup 2^{A_k}$, hence the output cannot be any partition $P$ such that $B\in P$ for some
$B\in2^N\backslash(2^{A_1}\cup\cdots\cup2^{A_k})$. In particular, if the input $\textbf q(0)=\textbf p$ is a partition $P=\{A_1,\ldots,A_{|P|}\}$, then the algorithm
only checks if $\textbf p$ is a local optimum, and updates if necessary. In the bottom case $\textbf q(0)=\textbf p_{\bot}$ the output $\textbf p^*=\textbf p_{\bot}$ coincides
with such an input (independently from input $\mu^v$). A seemingly general and flexible manner to choose the initial $n$ membership distributions is the following. Let
$\hat v(A)=\frac{v(A)}{|A|}$ and consider setting $\textbf q(0)$ via an arbitrary threshold $\theta\geq 0$ by:
\begin{equation}
q_i^A(0)=\left\{\begin{array}{c}0\text{ if }\hat v(A)\leq\theta\\ \hat v(A)\Big/\sum_{B\in 2^N_i:\hat v(B)>\theta}\hat v(B)\text{ otherwise}\end{array}\right.
\end{equation}
for all $i\in A$ and all $A\in 2^N$, entailing $\frac{q_i^A(0)}{q_i^B(0)}=\frac{\hat v(A)}{\hat v(B)}$ for all $i\in N$ and all $A,B\in2^N_i$ such that
$\hat v(A)>\theta<\hat v(B)$.

\subsection{Local search}
The greedy local-search strategy just described formally is:\smallskip\\
\textsl{GreedyClustering}$(w,\textbf q)$\smallskip\\
\textsl{Initialize:} Set $t=0$ and $\textbf q(0)$ as in expression (10).\smallskip\\
\textsl{GreedyLoop:} While $0<\sum_{i\in A}q^A_i(t)<|A|$ for some $A\in 2^N$, set $t=t+1$ and\smallskip\\
(a) select (randomizing in case of tails) one such $A=A^*(t)$ where for all\footnote{As usual colon ``:'' stands for ``such that''.} $B:0<\sum_{i\in B}q^B_j(t)<|B|$ average
derivative $\bar F^V_A$ satisfies $\bar F^V_A(\textbf q(t-1))\geq\bar F^V_B(\textbf q(t-1))$;\smallskip\\
(b) set $q_i^A(t)=
\left\{\begin{array}{c}1\text{ if }A=A^*(t)\\
0\text{ if }A\neq A^*(t)\end{array}\right.$ for all $i\in A^*(t),A\in 2^N_i$; \smallskip\\
(c) for all $j\in N\backslash A^*(t)$ and all $A\in 2^N_j:A\cap A^*(t)=\emptyset$, set $q^A_j(t)=q_j^A(t-1)+$
\begin{equation*}
+\left(\hat v(A)\sum_{\underset{B\cap A^*(t)\neq\emptyset}{B\in 2^N_j}}q_j^B(t-1)\right)
\left(\sum_{\underset{B'\cap A^*(t)=\emptyset}{B'\in 2^N_j}}\hat v(B')\right)^{-1}\text ;
\end{equation*}
(d) set $q^A_j(t)=0$ for all $j\in N\backslash A^*(t)$ and all $A\in 2^N_j:A\cap A^*(t)\neq\emptyset$.\smallskip\\
\textsl{CheckLoop:} While $q^A(t)=\chi_A$ and $v(A)<v(\{i\})+v(A\backslash i)$ for some $A\in 2^N,i\in A$, set $t=t+1$ and
\begin{eqnarray*}
q^{\hat A}_i(t)&=&\left\{\begin{array}{c}1\text{ if }|\hat A|=1\\
0\text{ otherwise}\end{array}\right .\text{ for all }\hat A\in 2^N_i\text ,\\
q^{B}_j(t)&=&\left\{\begin{array}{c}1\text{ if }B=A\backslash i\\
0\text{ otherwise}\end{array}\right .\text{ for all }j\in A\backslash i,B\in 2^N_j\text ,\\
q^{\hat B}_{j'}(t)&=&q^{\hat B}_{j'}(t-1)\text{ for all }j'\in A^c,\hat B\in 2^N_{j'}\text .
\end{eqnarray*}
\textsl{Output:} Set $\textbf p^*=\textbf q(t)$.

\begin{proposition}
The output $\textbf p^*$ of \textsl{GreedyClustering} is a local optimum: $F^V(\textbf p^*)\geq F^V(\textbf q)$ for all $\textbf q\in\mathcal N(\textbf p^*)$.
\end{proposition}

\begin{proof}
The case of singleton blocks $\{i\}$, if any, is trivial, in that $p_j^{*A}=0$ for all $j\neq i$ and all $A\in(2^N_j\cap2^N_i)$ entails
$F^V(q_i|\textbf p^*_{-i})=F^V(\textbf p^*)$ for any distribution $q_i\in\Delta_i$. Hence let $i\in A\in P^*$ with $|A|>1$. By switching from $p^*_i$
to $q_i\in\Delta_i$, global score variation is
$$F^V(q_i|\textbf p^*_{-i})-F^V(\textbf p^*)=v(\{i\})-v(A)+$$
$$+\left(q_i^A\sum_{B\in 2^A\backslash 2^{A\backslash i}:|B|>1}\mu^v(B)+\sum_{B'\in 2^{A\backslash i}}\mu^v(B')\right)=$$
$$=(q_i^A-1)\sum_{B\in 2^A\backslash 2^{A\backslash i}:|B|>1}\mu^v(B),$$
where the last equality is due to $$v(A)-v(A\backslash i)=\sum_{B\in 2^A\backslash2^{A\backslash i}}\mu^v(B).$$
Now assume that $F^V(q_i|\textbf p^*_{-i})-F^V(\textbf p^*)>0$, i.e. $\textbf p^*$ is not a local optimum. Since $q_i^A-1<0$ (because $q_i\neq p^*_i$), then
\begin{equation*}
\sum_{B\in 2^A\backslash 2^{A\backslash i}:|B|>1}\mu^v(B)=v(A)-v(A\backslash i)-v(\{i\})<0\text ,
\end{equation*}
and \textsl{CheckLoop} is meant precisely to screen out this.
\end{proof}

For the $\frac{n}{2}$-regular graphs seen above where in the worst-case \textsl{GreedyMerging} provides a null modularity score, it may be observed how
\textsl{GreedyClustering} surely (and immediately, in terms of the number of iterations) finds the unique optimum, given a reasonable input $\textbf q(0)$. Consider for
simplicity the initial $n$-tuple of uniform distributions, namely $q_i^A(0)=2^{1-n}$ for all $A\in 2^N_i,i\in N$. Then on every edge $\{i_k,j_k\}\in E$ where $i_k\in N^1$ and
$j_k\in N^2$ the average derivative takes value $\bar F^{\mathcal Q}_{\{i_k,j_k\}}(\textbf q(0))=\frac{1}{2}\left[-\frac{2}{n^2}+\frac{2}{2^{n-1}}\frac{2}{n^2}\right]=$\\
$=-\frac{1}{n^2}\left(1-\frac{1}{2^{n-2}}\right)<0$, while on every subset $A\subseteq N^1$ or $A\subseteq N^2$ its value is
$\bar F^{\mathcal Q}_{A}(\textbf q(0))=\frac{1}{|A|}\left[-\frac{|A|}{n^2}+\binom{|A|}{2}\frac{4}{n^22^{n-1}}\right]=$
$=-\frac{1}{n^2}\left(1-\frac{|A|-1}{2^{n-2}}\right)<0$, where $F^{\mathcal Q}$ is the MLE of additive partition function $\mathcal Q$ (i.e. modularity). Hence for $|A|=2$ there is
no difference, but $\bar F^{\mathcal Q}_A(\textbf q(0))$ increases with $|A|$ up to
\begin{equation*}
\bar F^{\mathcal Q}_{N^1}(\textbf q(0))=-\frac{1}{n^2}\left(1-\frac{\frac{n}{2}-1}{2^{n-2}}\right)=\bar F^{\mathcal Q}_{N^2}(\textbf q(0))\text ,
\end{equation*}
and $\frac{1}{2^{n-2}}<\frac{n-2}{2^{n-1}}$ (as $n>4$). Similar results obtain for the two quadratic scores defined by expressions (5-6) and (7-8). 

When $n$ is large, an input of uniform distributions clearly is not viable, in itself (consisting of $n2^{n-1}$ reals) and mostly because of the computational burden at
each iteartion. On the other hand, if the initial distributions place membership uniformly on pairs only, i.e.
$q_i^A(0)=\left\{\begin{array}{c}(n-1)^{-1}\text{ if }|A|=2\\ 0\text{ otherwise}\end{array}\right.$ for all $i\in N,A\in 2^N_i$, then the search cannot see that some triples
of vertices (namely those included in $N^1$ or in $N^2$) provide greater score. In fact, not only any two adjacent vertices provide the same score, but with this input
\textsl{GreedyClustering} is actually prevented from outputting any partition with (some) blocks larger than pairs. In other terms, the worst-case output of
\textsl{GreedyMerging} becomes certain. Also note that if scores are assigned according to expressions (7-8), then those $2\binom{\frac{n}{2}}{2}$ pairs included in $N^1$ or
in $N^2$ are more valuable than edges $\{i,j\},i\in N^1,j\in N^2$, because of common neighbors. However, with initial memberships distributed only on pairs
the algorithm remains constrained to partition the vertices into blocks no larger than pairs. These observations aim to highlight the crucial role played by locality
in the proposed search: if the initial distributions are too dispersed then the search is computationally impossible, but if they are too concentrated, especially over small
subsets, then the search space is too small. How to exploit this trade-off toward overlapping module detection is discussed hereafter.  

\subsection{Overlapping and multiple runs}
Local-search algorithms are commonly employed by varying the initial candidate solution over multiple runs. The idea is simple: eventually, among the resulting
multiple outputs a best one is chosen. Now, multiple outputs of \textsl{GreedyClustering}, with varying initial membership distributions, \textit{collectively} constitute a
family $\mathcal F\subset 2^N$ of overlapping vertex subsets. In other terms, the union of $k>1$ (optimal) partitions does provide the sought overlapping modular
structure. Formally, if $\textbf q_1(0),\ldots,\textbf q_k(0)$ are different initial fuzzy clusterings or inputs, with outputs $\textbf p^*_1,\ldots,\textbf p^*_k$ corresponding
to partitions $P^*_1,\ldots,P^*_k\in\mathcal P^N$, then $\mathcal F=P^*_1\cup\cdots\cup P^*_k$. In fact, the only case where $\mathcal F$ displays no overlapping is when these
outputs coincide, i.e. $P^*_l=P^*_{l+1},1\leq l<k$. Otherwise, their scores $V(P^*_l)=\sum_{A\in P^*_l}v(A)$ shall be generally different, but for overlapping module detection
those outputs with lower score may still be valuable, precisely because set function $v$ measures the score of subsets. That is to say, optimal partitions scoring lower
than others in terms of additive partition function $V$ may have some blocks scoring very high in terms of $v$. Therefore, the union of only some optimal partitions, namely
those scoring higher than some threshold, might exclude very important modules. These reasonings lead to see that $\mathcal F$ is in fact a \textit{weighted family}, with
weights $v(A)$ on its members $A\in\mathcal F$ quantified by $v$. 

Given its computational demand for generic initial membership distributions, \textsl{GreedyClustering} may run $k$ times only if each input $\textbf q_1(0),\ldots,\textbf q_k(0)$
is non-generic, i.e. with all $n$ memberships distributed only on small subsets. This severely restricts the search space by constraining the output partitions to only have small
blocks (see above). Then, family $\mathcal F$ only contains such small blocks, and thus important large modules might be excluded. On the other hand, in an overlapping structure
a large module seems likely to include some small ones. Accordingly, the search for large modules may be restricted by concentrating the initial $n$ membership distributions on
those vertex subsets given by the union of family members $A\in\mathcal F$. In particular, let $\Omega(\mathcal F)\subset 2^N$ be the set system\footnote{
$\Omega(\mathcal F)$ is a generaliziation of the field $2^P$ of subsets generated by partitions $P$, where
$2^{P_{\bot}}=2^N$, while $2^{P^{\top}}=\{\emptyset,N\}$.}
$\Omega(\mathcal F)=\{B:B=A_1\cup\cdots\cup A_m,\mathcal F\ni A_1,\ldots, A_m,m>0\}$. This is the collection of all (non-empty) subsets of $N$ resulting from the union of
some (i.e. at least one) family members $A\in\mathcal F$. As already explained, the search performed by \textsl{GreedyClustering} is top-down, as optimal blocks can only
be found among $\supseteq$-maximal subsets $A\in 2^N$ where vertices $i\in A$ initially place strictly positive membership. Weighted family $\mathcal F$ consisting precisely of
small modules, large ones may be detected by initially distributing memberships only on large subsets $B\in\Omega(\mathcal F)$. A simplest way to do this is uniformly over those
with size $|B|>\vartheta$ exceeding a threshold $\vartheta$, i.e.
$q_i^B(0)=\left\{\begin{array}{c}|2^N_{\Omega_i}|^{-1}\text{ if }|B|>\vartheta\\0\text{ otherwise}\end{array}\right.$ for all $B\in(2^N_i\cap\Omega(\mathcal F))$ and all $i\in N$,
where $2^N_{\Omega_i}$ denotes intesection $2^N_i\cap\{B:B\in\Omega(\mathcal F),|B|>\vartheta\}$. However, in the spirit of expression (10), weights or scores $v(B)$ of these
$B\in\Omega(N)$ with size exceeding $\vartheta$ may be used to determine more suitable non-uniform distributions. In any case, when initial memberships are (non-trivially)
distributed over some large subsets the search is computationally more demanding.

\section{Conclusion}
This work details how to search for network modules by means of a recent approach to objective function-based clustering and set partitioning \cite{Rossi_FCTA_2015,Rossi2017},
which applies to any graph clustering problem whose optimal solutions are extremizers of an additive partition function, namely a function assigning to every partition of
vertices the sum over blocks of their cluster score. This score of vertex subsets is quantified by a pseudo-Boolean (set) function, which in particular is quadratic when the score
of any subset is determined solely by the scores of included singletons and pairs. Although network topology is interpreted mostly in terms of alternative quadratic cluster scores,
still here the quadratic form is an option and not a constraint, as cubic cluster scores appear as well, aimed at incorporating the clustering coefficient of spanned subgraphs into
the objective function. Modularity being indeed an additive partition function with quadratic cluster scores, the whole setting is thorough detailed for the well-known case of
modularity maximization \cite{CommunitiesReviewNewmannSIAM2003,NewmanPNAS2006,BookBarabasiNewmann2006,LocalModularityClustering2011}.
%

The optimization-based search for network modules is conceived in the continuous space of fuzzy clusterings, because the objective function is in fact the polynomial multilinear
extension of additive partition functions. The extremizers are shown to be partitions of nodes, hence the proposed greedy local-search procedure outputs a graph partition.
In particular, a greedy loop generates blocks by maximizing the average derivative, where this latter parallels the standard derivative of pseudo-Boolean functions
\cite{BorosHammer02}, but most importantly an input fuzzy clustering where to start from makes the search local. The choice of this initialization is crucial, as it balances
between
computational burden and search width. Suitable inputs might allow for multiple runs, firstly searching for small modules only, and secondly for large ones only. Then, the
outputs are partitions, and their union is a set system of vertex subsets, namely an overlapping modular structure.

\subsection{Future work}
How to conceive benchmark graphs for testing overlapping module detection methods seems far from obvious. The comparison between probabilistic models and real-world complex systems
is at the heart of network analysis, and such models generally rely on maintaining randomness insofar as possible, while fixing quantities such as the degree sequence and the
clustering coefficient \cite{MassiveSyntheticNetworksArxiv2017,LancichinettiEtAl2008,CommunitiesReviewNewmannSIAM2003,NewmanRandomTriangles2009}. For random networks with `fixed'
modular
structure, a problem is that there is no precise quantity to fix, namely no measurement of modules in real-world networks. One way to deal with this is to fix a partition $P$ of
$N$ and
the corresponding
partition-like graph $G_P=(N,E_P)=\cup_{A\in P}K_A$ (each of whose components is the complete graph $K_A$ on a block $A\in P$, see Section V),
and next introduce `noise' by randomly adding some edges
$\{i',j'\}\in N_2\backslash E_P$ while removing some edges $\{i,j\}\in E_P$. Since blocks are non-overlapping, a random network $G$ with
\textit{overlapping} modules then may obtain by developing from the union of $k$ \textit{cliques}: $G=K_{A_1}\cup\cdots\cup K_{A_k}$. This means focusing on the \textit{clique-type}
$\kappa=(\kappa_1,\ldots,\kappa_n)$, i.e. the number $\kappa_m$ of maximal complete subgraphs on $m$ vertices, $1\leq m\leq n$. In fact, despite NP-hardness, recently such an information has become 
available even for very large networks \cite{CliqueEnumerationDistributedComputing2009,MaxCliqueEnum2017InfoProcLett}.
The $\frac{n}{2}$-regular graph $G$ considered in Section VI is the union  
$G=K_{N^1}\cup K_{N^2}\cup K_{\{i_1,j_1\}}\cup\cdots\cup K_{\{i_{\frac{n}{2}},j_{\frac{n}{2}}\}}$ of $\frac{n}{2}+2$ overlapping cliques, with clique-type $\kappa_2=\frac{n}{2},\kappa_{\frac{n}{2}}=2$ and
$\kappa_m=0$ for $2\neq m\neq\frac{n}{2}$. The clique-type $\kappa$ of graphs is a generalization of the type (or class) of partitions \cite{Rota64}, and if it is fixed,
then randomness essentially concerns the sizes of pair-wise intersections, i.e.
$|A\cap A'|$ for $K_A,K_{A'}\subset G$, which in turn
determine vertex degrees (and also seemingly provide the needed noise). More generally, a flexible probabilistic model might employ the clique-type observed in real-world networks as a parameter, rather than maintaining it fixed.

\bibliographystyle{abbrv}
\bibliography{biblioIJNN}

\end{document}